\documentclass[conference,letterpaper]{IEEEtran}
\IEEEoverridecommandlockouts

\usepackage{cite}
\usepackage{amsmath,amssymb,amsfonts,amsthm}
\usepackage{graphicx}
\usepackage{textcomp}
\usepackage{xcolor}
\usepackage{booktabs}
\usepackage{multirow}
\usepackage{array}
\usepackage{balance}
\usepackage[hyphens]{url}
\newtheorem{remark}{Remark}
\usepackage[
    bookmarks=false,
    pdfpagelabels=false,
    hypertexnames=false,
    hidelinks
]{hyperref}

\graphicspath{{Figures/}}

\newtheorem{theorem}{Theorem}
\newtheorem{lemma}{Lemma}
\newtheorem{proposition}{Proposition}

\begin{document}

\title{Rate-Splitting for Service-Aware Multi-Orbit Orchestration in 6G TN--NTN Networks
\thanks{This work is supported by the US--Ireland R\&D Partnership Programme Project ``Resilient Networks'' under Grant RI-SFI-23/US/3924, the EU MSCA Project ``COALESCE'' under Grant 101130739, and Research Ireland Grant 13/RC/2077\_P2.}
}

\author{
\IEEEauthorblockN{
Sayanti Ghosh\textsuperscript{1},
Mustafa Kishk\textsuperscript{2},
Nicola Marchetti\textsuperscript{1}
}
\IEEEauthorblockA{
\textsuperscript{1}Department of Electrical and Electronic Engineering, Trinity College Dublin, Ireland}
\textsuperscript{2} Department of Electronic Engineering, Maynooth University, Ireland
\IEEEauthorblockA{
Email: saghosh@tcd.ie, 
mustafa.kishk@mu.ie,
nicola.marchetti@tcd.ie}
}

\maketitle

\begin{abstract}
Integrated terrestrial--non-terrestrial networks (TN--NTNs) are a key
enabler of sixth-generation (6G) connectivity across terrestrial,
geostationary Earth orbit (GEO), medium Earth orbit (MEO), and low
Earth orbit (LEO) domains. However, existing NTN frameworks do not
jointly address service-aware serving-domain selection, multi-orbit
coordination, and hierarchical multi-stream transmission. This paper
proposes a service-aware multi-orbit rate-splitting multiple access
(SA-MO-RSMA) framework that integrates global-common, orbit-common,
and user-private streams with service-intent-aware serving-domain
selection and resource orchestration. The framework operates as a
system-level orchestration layer over existing 3GPP NTN architectures.
A multi-criteria Orbit--Service Utility Function (OSUF) is introduced,
together with metrics for service continuity, switching stability,
power-normalized transmission efficiency, and synchronization
robustness. Numerical results show that SA-MO-RSMA improves average
per-user service utility, service continuity, switching stability, and
transmission efficiency over the considered benchmark schemes,
demonstrating its potential for coordinated service-aware multi-orbit
TN--NTN operation.
\end{abstract}

\section{Introduction}
\label{sec:introduction}

Sixth-generation (6G) networks are envisioned to integrate terrestrial
and non-terrestrial networks (TN--NTNs) to provide wide-area coverage,
resilient connectivity, and heterogeneous services across terrestrial,
geostationary Earth orbit (GEO), medium Earth orbit (MEO), and low Earth
orbit (LEO) domains~\cite{IMT2030,TR38811,TR38821}. Existing 3GPP NTN
studies address propagation delay, mobility, channel characteristics,
and protocol adaptation~\cite{TR38811,TR38821}, while multi-layer NTN
architectures have been investigated for coordinated operation across
heterogeneous orbital platforms~\cite{Araniti2022NTN,Guidotti2024NTN}.
Recent work has also considered resource allocation in integrated
satellite--terrestrial networks~\cite{Koutsioumpa2025TVT}.

In parallel, rate-splitting multiple access (RSMA) provides flexible
interference management by splitting user messages into common and
private streams~\cite{MaoRSMA2018,ClerckxRSMA2023}. Recent studies have
extended RSMA to GEO--LEO coexisting satellite systems
\cite{Ryu2024TVT}, integrated satellite--terrestrial networks
\cite{Hu2025TVT}, and multibeam LEO satellite networks
\cite{Seong2026TVT}. These works demonstrate the potential of RSMA for
satellite communications, but mainly focus on physical-layer
transmission, precoding, outage, or throughput optimization. The joint
design of service-aware serving-domain selection, hierarchical
multi-orbit RSMA transmission, service continuity, switching stability,
synchronization robustness, and utility-driven resource orchestration
therefore remains insufficiently investigated in integrated TN--NTN
systems.

To address this gap, we propose service-aware multi-orbit rate-splitting multiple access (SA-MO-RSMA), a service-aware multi-orbit
framework integrating hierarchical RSMA transmission,
service-intent-aware serving-domain selection, and utility-driven
resource orchestration. The main contributions are:
\begin{itemize}
\item We develop a hierarchical multi-orbit RSMA model combining
global-common, orbit-common, and user-private streams for coordinated
GEO/MEO/LEO/TN operation.

\item We propose a service-intent-aware serving-domain selection
mechanism accounting for latency, reliability, positioning, multicast,
synchronization, and communication-quality requirements.

\item We introduce the Orbit--Service Utility Function (OSUF),
Multi-Orbit Service Continuity Ratio (MOSCR), Orbit Switching Stability
Index (OSSI), Orbit-Aware RSMA Efficiency (OARE), and synchronization
quality metric $J$ to quantify service utility, continuity, switching
stability, power-normalized transmission efficiency, and synchronization
robustness, respectively.

\item We formulate the joint serving-domain selection and
resource-allocation problem and evaluate SA-MO-RSMA against conventional
RSMA, OMA/NOMA, and benchmark serving-domain selection schemes under
3GPP-aligned NTN settings.
\end{itemize}

SA-MO-RSMA performs service-aware serving-domain selection through a
utility-based orchestration layer and subsequently coordinates
global-common, orbit-common, and user-private RSMA streams. It operates
as a system-level orchestration framework over the existing 3GPP
NR--NTN architecture while preserving the adopted NTN propagation and
radio-access assumptions.
\section{System Model}
\label{sec:system_model}
\begin{figure}[!t]
\centering
\includegraphics[width=0.48\textwidth]{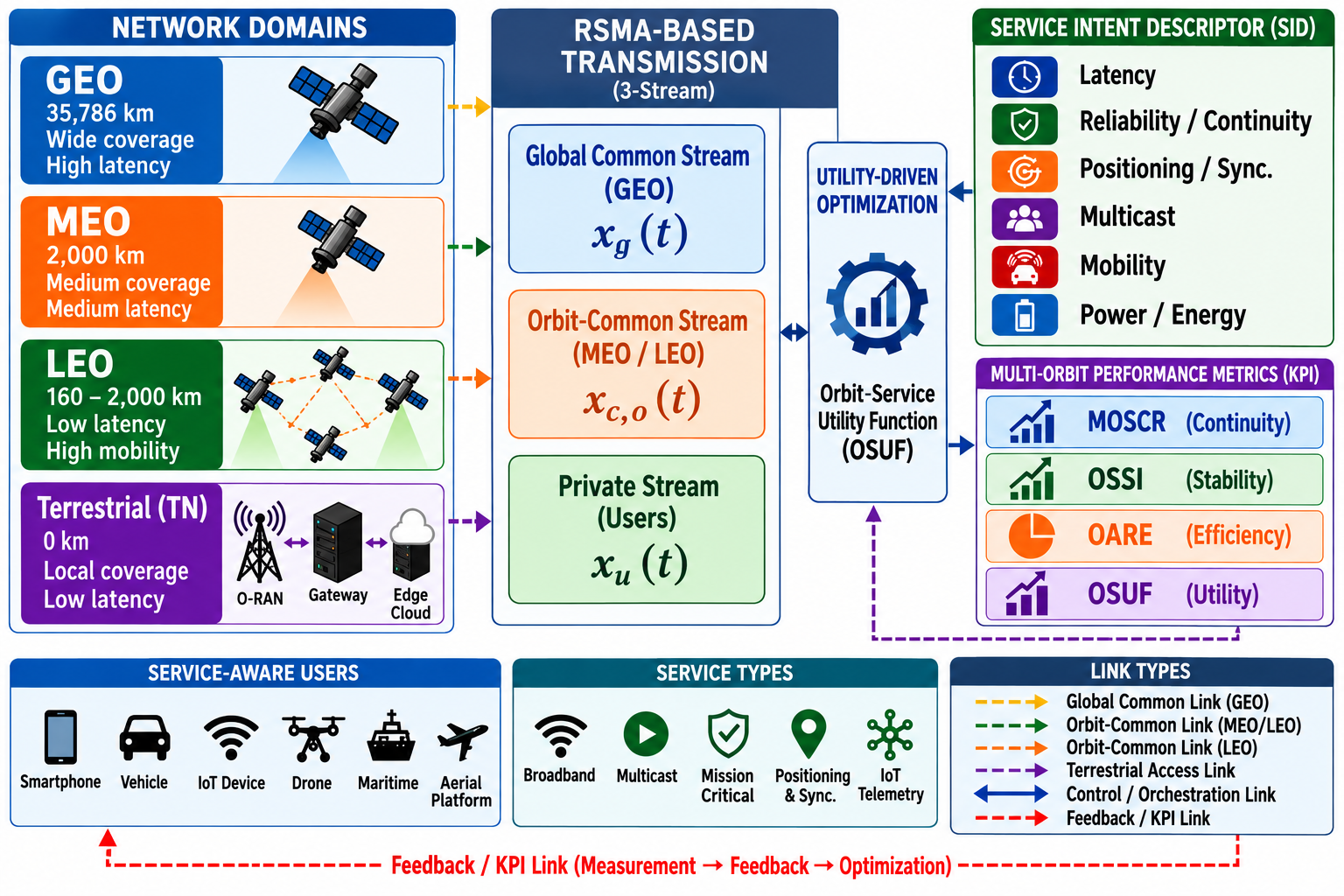}
\caption{SA-MO-RSMA architecture for service-aware serving-domain
selection and hierarchical RSMA transmission across integrated
GEO/MEO/LEO/TN domains.}
\label{fig:sa_mo_rsma_architecture}
\end{figure}

We consider an integrated TN--NTN comprising terrestrial, GEO,
MEO, and LEO domains, as shown
in Fig.~\ref{fig:sa_mo_rsma_architecture}. SA-MO-RSMA operates through
service-intent-aware serving-domain selection, hierarchical RSMA
transmission, and performance-feedback-driven orchestration. Let
$\mathcal{O}=\{\mathcal{O}_G,\mathcal{O}_M,\mathcal{O}_L,\mathcal{O}_T\}$
denote the GEO, MEO, LEO, and terrestrial domains, respectively, and
$\mathcal{U}=\{1,\ldots,U\}$ the user set. Each user is associated with
one serving domain per scheduling interval, with
$o_u\in\mathcal{O}$ denoting the domain serving user $u$.

For user $u$, the Service Intent Descriptor (SID) is
$\mathcal{S}_u=\{L_u,R_u,P_u^{\mathrm{pos}},M_u\}$, where the elements
denote the latency, reliability, positioning-accuracy, and multicast
requirements, respectively. The end-to-end latency under domain $o$ is
$L_{u,o}=d_{u,o}/c+L_{u,o}^{\mathrm{proc}}
+L_{u,o}^{\mathrm{queue}}+L_{u,o}^{\mathrm{sync}}
+L_{u,o}^{\mathrm{orch}}$, where $d_{u,o}$ is the propagation distance,
$c$ is the speed of light, and the remaining terms represent processing,
queueing, synchronization, and orchestration delays.

Following the RSMA downlink principle
\cite{MaoRSMA2018,ClerckxRSMA2023}, the transmitted signal is
$\mathbf{x}(t)=\mathbf{p}_g s_g(t)
+\sum_{o\in\mathcal{O}}\mathbf{p}_{c,o}s_{c,o}(t)
+\sum_{u\in\mathcal{U}}\mathbf{p}_u s_u(t)$,
where $s_g(t)$, $s_{c,o}(t)$, and $s_u(t)$ are the unit-power
global-common, orbit-common, and user-private symbols, and
$\mathbf{p}_g$, $\mathbf{p}_{c,o}$, and $\mathbf{p}_u$ are the
corresponding precoders. Their allocated powers are
$P_g=\|\mathbf{p}_g\|^2$, $P_{c,o}=\|\mathbf{p}_{c,o}\|^2$, and
$P_u^{\mathrm{tx}}=\|\mathbf{p}_u\|^2$, satisfying
$P_g+\sum_{o\in\mathcal{O}}P_{c,o}
+\sum_{u\in\mathcal{U}}P_u^{\mathrm{tx}}\leq P_{\max}$.

Let $\mathbf{h}_{u,o}$ denote the channel vector from domain $o$ to
user $u$, incorporating the adopted NTN propagation effects
\cite{TR38811}. The received signal at user $u$ is
\begin{equation}
\begin{aligned}
y_u(t)=&
\mathbf{h}_{u,o_u}^{H}\mathbf{p}_g s_g(t)
+\sum_{o\in\mathcal{O}}
\mathbf{h}_{u,o}^{H}\mathbf{p}_{c,o}s_{c,o}(t)\\
&+\sum_{v\in\mathcal{U}}
\mathbf{h}_{u,o_v}^{H}\mathbf{p}_{v}s_v(t)+n_u(t),
\end{aligned}
\label{eq:received_signal}
\end{equation}
where $o_v$ is the serving domain of user $v$ and $n_u(t)$ denotes
additive receiver noise. The effective term
$\mathbf{h}_{u,o}^{H}\mathbf{p}_j$ captures channel attenuation and
precoding/beamforming gain; explicit beamformer design is outside the
scope of this system-level orchestration study.

Each user successively decodes the global-common stream, the
orbit-common stream associated with its serving domain, and its
user-private stream. At each stage, undecoded streams are treated as
interference, while successfully decoded common streams are removed
through successive interference cancellation (SIC).

To enable service-aware serving-domain selection, the utility of
assigning user $u$ to candidate domain $o$ is defined as
\begin{equation}
\Psi_{u,o}
=w_1\hat{\Gamma}_{u,o}+w_2\hat{C}_{u,o}
+w_3\hat{P}_{u,o}+w_4\hat{M}_{u,o}
+w_5\hat{J}_{u,o}-w_6\hat{L}_{u,o},
\label{eq:orbit_utility}
\end{equation}
where $\hat{\Gamma}_{u,o}$, $\hat{C}_{u,o}$, $\hat{P}_{u,o}$,
$\hat{M}_{u,o}$, $\hat{J}_{u,o}$, and $\hat{L}_{u,o}$ denote normalized
communication quality, service continuity, positioning capability,
multicast efficiency, synchronization stability, and latency cost,
respectively. Each component is normalized to $[0,1]$ as
$\hat{X}=(X-X_{\min})/(X_{\max}-X_{\min})$ over the candidate domains,
with $w_i\geq0$ and $\sum_{i=1}^{6}w_i=1$.

The serving domain is selected as
\begin{equation}
o_u^\star=\arg\max_{o\in\mathcal{O}}\Psi_{u,o},
\label{eq:orbit_selection}
\end{equation}
which determines the subsequent RSMA stream association, scheduling,
and domain-dependent transmission configuration.

\begin{lemma}
If all normalized components in eq.~\eqref{eq:orbit_utility} lie in
$[0,1]$, with $w_i\geq0$ and $\sum_{i=1}^{6}w_i=1$, then
$\Psi_{u,o}$ is bounded.
\end{lemma}

\begin{proposition}
For every $u\in\mathcal{U}$, the selection rule in
eq.~\eqref{eq:orbit_selection} admits at least one maximizer over the
finite set $\mathcal{O}$.
\end{proposition}

Proofs are provided in the Appendix.
\section{Performance Metrics and Optimization Framework}
\label{sec:performance}

To evaluate SA-MO-RSMA, we define the OSUF, which extends eq.~\eqref{eq:orbit_utility} to jointly capture
normalized throughput, service continuity, positioning reliability,
multicast efficiency, synchronization stability, latency cost, and
orchestration overhead:
\begin{equation}
\begin{aligned}
U_{u,o}^{\mathrm{OSUF}}
={}&w_T\hat{T}_{u,o}+w_C\hat{C}_{u,o}
+w_P\hat{P}_{u,o}+w_M\hat{M}_{u,o}\\
&+w_J\hat{J}_{u,o}-w_L\hat{L}_{u,o}
-w_O\hat{O}_{u,o},
\end{aligned}
\label{eq:osuf}
\end{equation}
All components lie in $[0,1]$, with non-negative weights satisfying
$w_T+w_C+w_P+w_M+w_J+w_L+w_O=1$. The network utility is
$U_{\mathrm{net}}^{\mathrm{OSUF}}
=\sum_{u=1}^{U}U_{u,o_u}^{\mathrm{OSUF}}$.

Service continuity is quantified by
$\mathrm{MOSCR}_u=T_{u,\mathrm{connected}}/T_{u,\mathrm{total}}$
and $\mathrm{MOSCR}=U^{-1}\sum_{u=1}^{U}\mathrm{MOSCR}_u$, where
$T_{u,\mathrm{connected}}$ and $T_{u,\mathrm{total}}$ denote the
connected-service duration and total observation time, respectively.

To characterize serving-domain switching stability, we define the
OSSI as
\begin{equation}
\mathrm{OSSI}=
\begin{cases}
1-\dfrac{N_{\mathrm{unstable}}}{N_{\mathrm{switch}}},
& N_{\mathrm{switch}}>0,\\[1mm]
1, & N_{\mathrm{switch}}=0.
\end{cases}
\label{eq:ossi}
\end{equation}
Here, $N_{\mathrm{unstable}}$ counts ping-pong transitions
$o_a\!\rightarrow\!o_b\!\rightarrow\!o_a$ occurring within the
predefined observation window $T_{\mathrm{pp}}$, and
$N_{\mathrm{switch}}$ is the total number of serving-domain switches.
The second case assigns maximum stability when no switching occurs.

To quantify the transmit-power efficiency of hierarchical RSMA, the
OARE is defined as
\begin{equation}
\mathrm{OARE}=
\frac{R_g+\sum_{o\in\mathcal O}R_{c,o}
+\sum_{u\in\mathcal U}R_{u,p}}
{P_g+\sum_{o\in\mathcal O}P_{c,o}
+\sum_{u\in\mathcal U}P_u^{\mathrm{tx}}},
\label{eq:oare}
\end{equation}
where $R_g$, $R_{c,o}$, and $R_{u,p}$ denote the achievable
global-common, orbit-common, and user-private stream rates,
respectively.
Synchronization robustness is characterized by
$\Delta_o=\sqrt{(\sigma_{\tau,o}/\tau_{\rm ref})^2+
(\sigma_{f,o}/f_{\rm ref})^2}$,
$J_o=(1+\Delta_o)^{-1}$, and
$J=|\mathcal O|^{-1}\sum_{o\in\mathcal O}J_o$, where
$\sigma_{\tau,o}$ and $\sigma_{f,o}$ denote the timing- and
frequency-error standard deviations of domain $o$, respectively, and
$\tau_{\rm ref}$ and $f_{\rm ref}$ are the corresponding normalization
reference levels. Thus, $\Delta_o$ is a dimensionless synchronization
impairment, $J_o\in(0,1]$ denotes the synchronization quality of domain
$o$, and $J$ represents the average synchronization quality across the
candidate domains. The normalized domain-dependent quantity
$\hat{J}_{u,o}$ is used in the OSUF in eq.~\eqref{eq:osuf}.

\begin{lemma}
If all normalized OSUF components lie in $[0,1]$, with $w_k\geq0$
and $\sum_k w_k=1$, then $U_{u,o}^{\mathrm{OSUF}}$ is bounded.
\end{lemma}
The proof is provided in the Appendix.

\begin{remark}
Since $0\leq N_{\mathrm{unstable}}\leq N_{\mathrm{switch}}$ for
$N_{\mathrm{switch}}>0$, eq.~\eqref{eq:ossi} guarantees
$0\leq\mathrm{OSSI}\leq1$; when $N_{\mathrm{switch}}=0$,
$\mathrm{OSSI}=1$. Moreover, the synchronization metric satisfies
$0<J\leq1$.
\end{remark}
\begin{proposition}
For fixed $\hat{T}_{u,o}$, $\hat{C}_{u,o}$, $\hat{P}_{u,o}$,
$\hat{M}_{u,o}$, and $\hat{J}_{u,o}$,
$U_{u,o}^{\mathrm{OSUF}}$ is non-increasing in
$\hat{L}_{u,o}$ and $\hat{O}_{u,o}$.
\end{proposition}
The proof is provided in the Appendix.

Using these metrics, the joint serving-domain selection and
resource-allocation problem is formulated as
\begin{align}
\max_{\mathbf P,\mathbf A,\mathbf\Theta}\quad&
U_{\mathrm{net}}^{\mathrm{OSUF}}
\label{eq:opt_obj}\\[-1mm]
\textnormal{s.t.}\quad&
\mathrm{MOSCR}\geq C_{\mathrm{th}},\quad
\mathrm{OSSI}\geq S_{\mathrm{th}},
\nonumber\\[-1mm]
&
\mathrm{OARE}\geq E_{\mathrm{th}},\quad
J\geq J_{\mathrm{th}},
\label{eq:const_metrics}\\[-1mm]
&
P_g+\sum_{o\in\mathcal O}P_{c,o}
+\sum_{u=1}^{U}P_u^{\mathrm{tx}}\leq P_{\max},
\label{eq:const_power}\\[-1mm]
&
o_u\in\mathcal O,\quad
P_g,P_{c,o},P_u^{\mathrm{tx}}\geq0.
\label{eq:const_nonnegative}
\end{align}
Here, $\mathbf P$ collects the stream-power variables,
$\mathbf A=\{o_u\}_{u\in\mathcal U}$ denotes the serving-domain
associations, and $\mathbf\Theta$ collects the remaining RSMA,
scheduling, and orchestration variables. The thresholds
$C_{\mathrm{th}}$, $S_{\mathrm{th}}$, $E_{\mathrm{th}}$, and
$J_{\mathrm{th}}$ specify the minimum service-continuity,
switching-stability, transmission-efficiency, and synchronization-quality
requirements, respectively. The problem is mixed discrete--continuous
and generally non-convex due to coupling between serving-domain
associations and continuous resource variables.

We employ an alternating service-aware optimization procedure.
Equation~\eqref{eq:orbit_selection} initializes $\mathbf A$. For fixed
$\mathbf A$, $(\mathbf P,\mathbf\Theta)$ are updated to increase
$U_{\mathrm{net}}^{\mathrm{OSUF}}$ subject to
eqs.~\eqref{eq:const_metrics}--\eqref{eq:const_nonnegative}, followed
by an update of $\mathbf A$ using the resulting utilities. The procedure
terminates when $\mathbf A$ remains unchanged or the OSUF improvement
falls below the convergence tolerance $\epsilon_{\mathrm{opt}}>0$.
Since the problem is non-convex, the procedure yields a feasible
locally optimized solution without claiming global optimality. Over
$I_{\max}$ iterations, the association update requires at most
$I_{\max}U|\mathcal O|$ utility evaluations, i.e., $4I_{\max}U$ for
$|\mathcal O|=4$; the overall complexity additionally depends on the
solver used for the continuous resource-update step.

\begin{theorem}
If the feasible set in
eqs.~\eqref{eq:const_metrics}--\eqref{eq:const_nonnegative} is
nonempty and compact, and $U_{\mathrm{net}}^{\mathrm{OSUF}}$ is
continuous in the continuous variables for each fixed association,
then eq.~\eqref{eq:opt_obj} admits at least one global maximizer.
\end{theorem}
The proof is provided in the Appendix.

For load-dependent analysis, define the average per-user OSUF as
$\bar{U}_{\mathrm{net}}^{\mathrm{OSUF}}
=U^{-1}\sum_{u=1}^{U}U_{u,o_u}^{\mathrm{OSUF}}$.
\begin{theorem}
Under fixed total resources and channel conditions, if increasing user
load does not increase the per-user benefit components and does not
decrease the corresponding cost components, then
$\bar{U}_{\mathrm{net}}^{\mathrm{OSUF}}$ is non-increasing with user
load.
\end{theorem}
The proof is provided in the Appendix.

\section{Results and Discussion}
\label{sec:results}
This section evaluates the proposed SA-MO-RSMA framework against
conventional RSMA, OMA/NOMA, and serving-domain selection baselines.
OSUF, MOSCR, OSSI~\eqref{eq:ossi}, and OARE~\eqref{eq:oare} evaluate
service utility, service continuity, switching stability, and
transmission efficiency, respectively. Simulations adopt 3GPP-aligned
NTN parameters~\cite{TR38811,TR38821} and IMT-2030
considerations~\cite{IMT2030}, with orbital altitudes and Ka-band
propagation settings based on representative NTN assumptions
\cite{TR38811,Araniti2022NTN}. Unless otherwise specified, 25\% of the
total transmit power is allocated to the global-common stream, while
the remaining 75\% is allocated between the orbit-common and
user-private streams through the resource-allocation procedure. The
serving-domain association is initialized using
eq.~\eqref{eq:orbit_selection} and the utility in
eq.~\eqref{eq:orbit_utility}, while the joint optimization enforces
service-continuity, switching-stability, transmission-efficiency, and
synchronization constraints. All reported feasible solutions satisfy
$J\geq J_{\rm th}$; hence, $J$ is treated as a feasibility requirement
rather than as a separately swept performance metric. The simulation
parameters are summarized in Table~\ref{tab:sim_parameters}.

\begin{table}[t]
\caption{Simulation Parameters}
\label{tab:sim_parameters}
\centering
\footnotesize
\renewcommand{\arraystretch}{0.95}
\setlength{\tabcolsep}{2.2pt}
\begin{tabular}{lll}
\toprule
\textbf{Parameter} & \textbf{Description} & \textbf{Value} \\
\midrule
$f_c$                & Carrier frequency                    & 20 GHz \\
$B$                  & System bandwidth                     & 400 MHz \\
$P_{\max}$           & Total transmit-power budget          & 40 dBm \\
$N_0$                & Noise power spectral density         & $-174$ dBm/Hz \\
$U$                  & Number of users                      & 10--80 \\
$v_u$                & User speed                           & 20--160 km/h \\
$h_{\rm LEO}$        & LEO altitude                         & 600 km \\
$h_{\rm MEO}$        & MEO altitude                         & 8,000 km \\
$h_{\rm GEO}$        & GEO altitude                         & 35,786 km \\
$\sigma_{\rm sh}$    & Shadow-fading std. dev. (urban LOS)  & 4 dB \\
$\alpha_g$           & Global-common power fraction         & 0.25 \\
$1-\alpha_g$         & Remaining RSMA power fraction        & 0.75 \\
$\mathbf{w}$         & OSUF weight vector                   & $(1/7,\ldots,1/7)$ \\
$C_{\rm th}$         & MOSCR threshold                      & 0.8 \\
$S_{\rm th}$         & OSSI threshold                       & 0.8 \\
$E_{\rm th}$         & OARE threshold                       & 0.5 \\
$J_{\rm th}$         & Synchronization threshold            & 0.8 \\
$T_{\rm pp}$         & Ping-pong observation window         & 1 s \\
$\tau_{\rm ref}$     & Timing-error reference               & 1 $\mu$s \\
$f_{\rm ref}$        & Frequency-error reference            & 1 kHz \\
$\epsilon_{\rm opt}$ & Optimization tolerance               & $10^{-3}$ \\
$I_{\max}$           & Maximum optimization iterations      & 50 \\
$N_{\rm MC}$         & Monte Carlo realizations             & $10^{4}$ \\
Channel model        & 3GPP NTN propagation model           & TR 38.811 \\
\bottomrule
\end{tabular}
\end{table}
Equal OSUF weights are used as a service-neutral reference setting;
service-specific weighting is left to application-dependent
orchestration policies.
For mobility-aware switching, a candidate domain is activated only when
its SINR exceeds that of the current domain by the prescribed
hysteresis margin; otherwise, the current association is retained.
For comparison, the max-SINR and latency-only baselines select the
domain with the highest instantaneous SINR and minimum latency,
respectively; fixed-domain selection retains the initial association,
while the no-hysteresis baseline switches without the hysteresis
condition. All schemes use identical channel, mobility, and power
settings.
\begin{figure}[t]
\centering
\includegraphics[width=0.82\columnwidth]{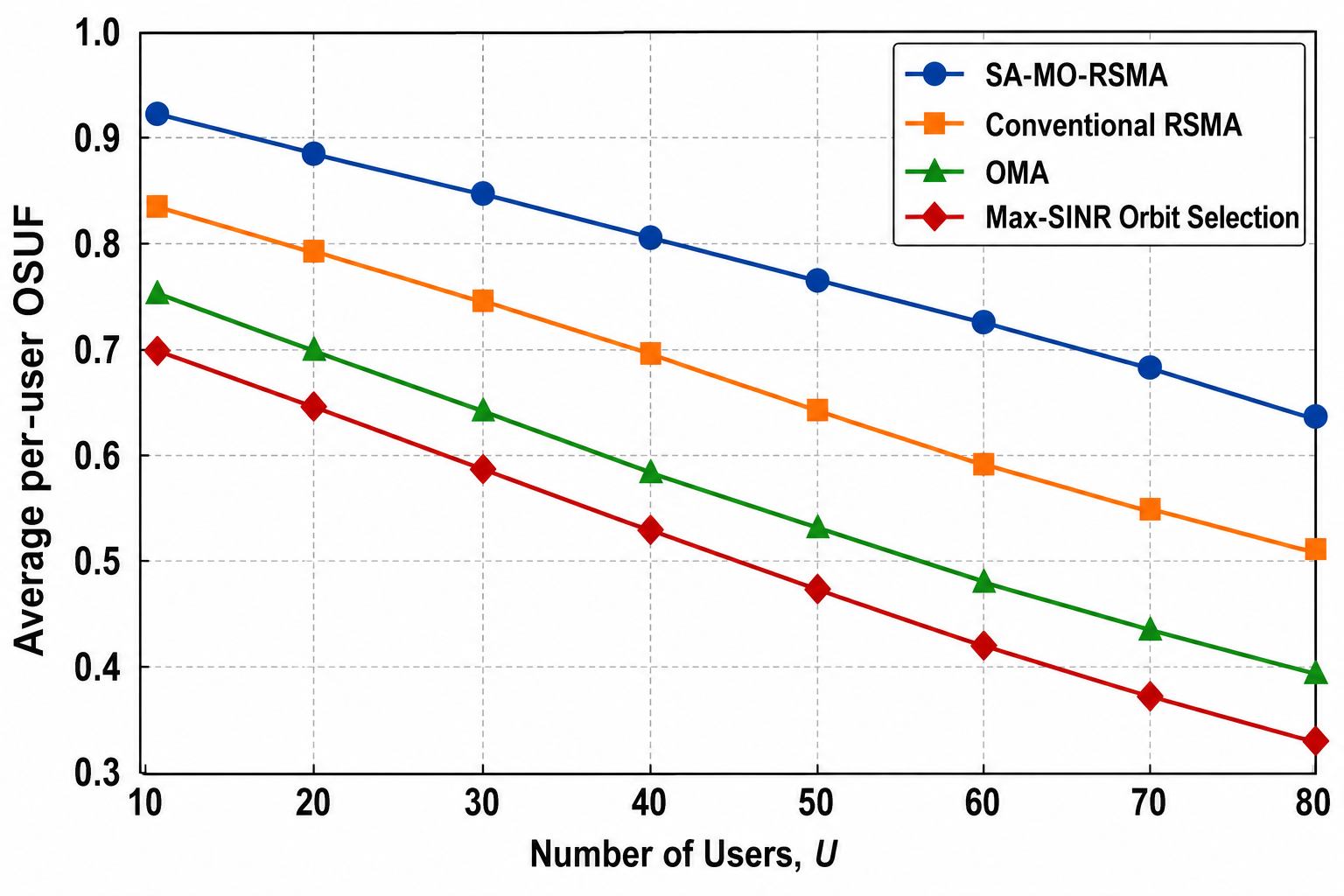}
\caption{Average per-user OSUF versus number of users.}
\label{fig:osuf_users}
\vspace{-3mm}
\end{figure}

Fig.~\ref{fig:osuf_users} shows the average per-user utility
$\bar{U}_{\mathrm{net}}^{\mathrm{OSUF}}$ versus the number of users.
The utility decreases with increasing user load over the considered
range, reflecting increased interference, resource contention, and
orchestration cost. SA-MO-RSMA achieves the highest utility across all
considered user loads. The gain over conventional RSMA reflects the
benefit of the proposed hierarchical stream structure, whereas OMA and
max-SINR selection are limited by orthogonal resource allocation and
purely channel-driven serving-domain decisions, respectively.

\begin{figure}[t]
\centering
\includegraphics[width=0.82\columnwidth]{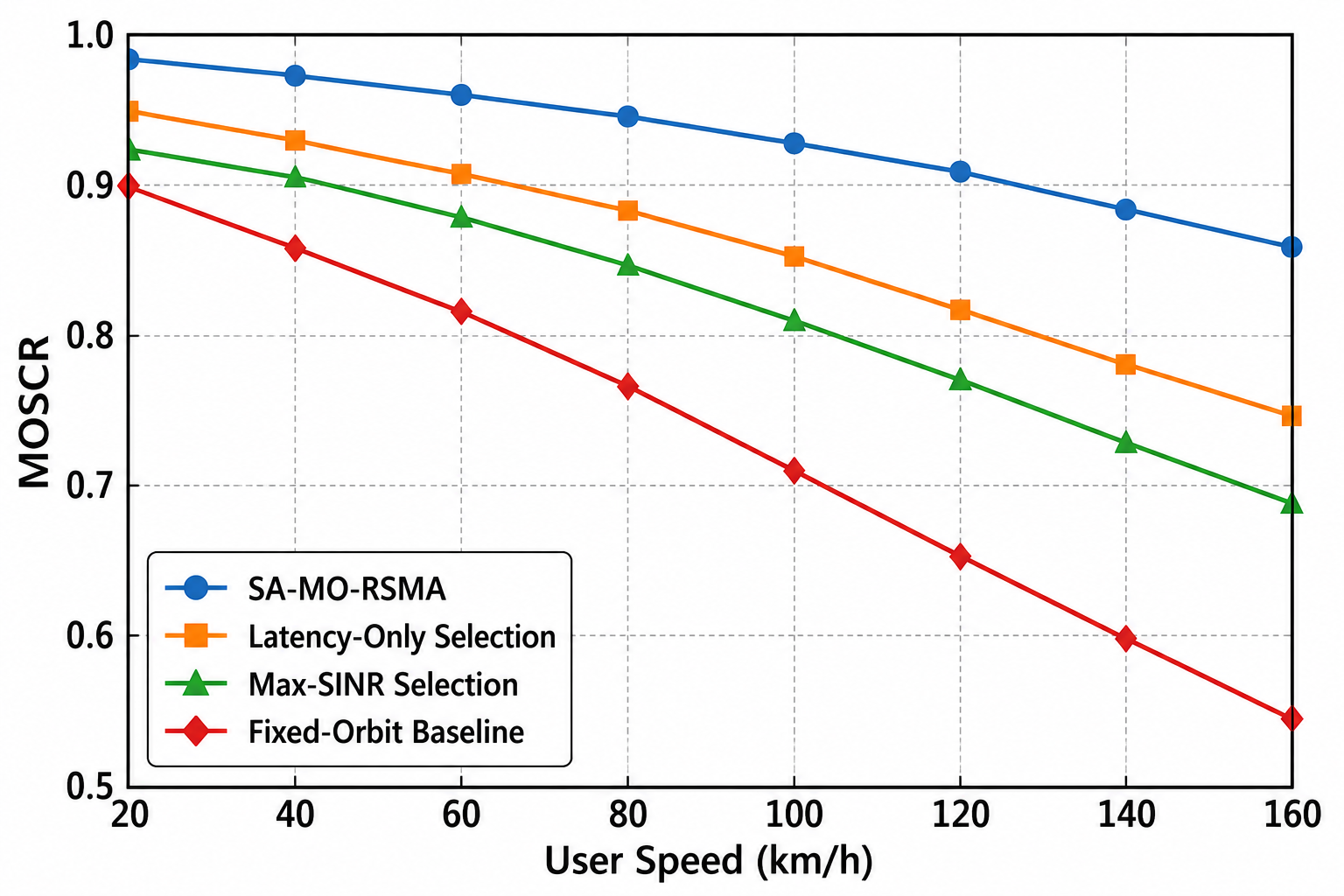}
\caption{MOSCR versus user speed.}
\label{fig:moscr_speed}
\vspace{-3mm}
\end{figure}

Fig.~\ref{fig:moscr_speed} shows MOSCR versus user speed. As user
mobility increases, MOSCR decreases for all schemes because faster
channel variations and more frequent serving-domain transitions reduce
uninterrupted connected time. SA-MO-RSMA achieves the highest MOSCR
across the considered speed range, demonstrating improved service
continuity under mobility. This gain results from its service-aware
serving-domain selection mechanism, which jointly accounts for
communication quality, service continuity, and orchestration
requirements. In contrast, latency-only and max-SINR baselines rely on
more limited selection criteria, while fixed-domain selection lacks the
adaptability required under changing mobility conditions.

\begin{figure}[t]
\centering
\includegraphics[width=0.82\columnwidth]{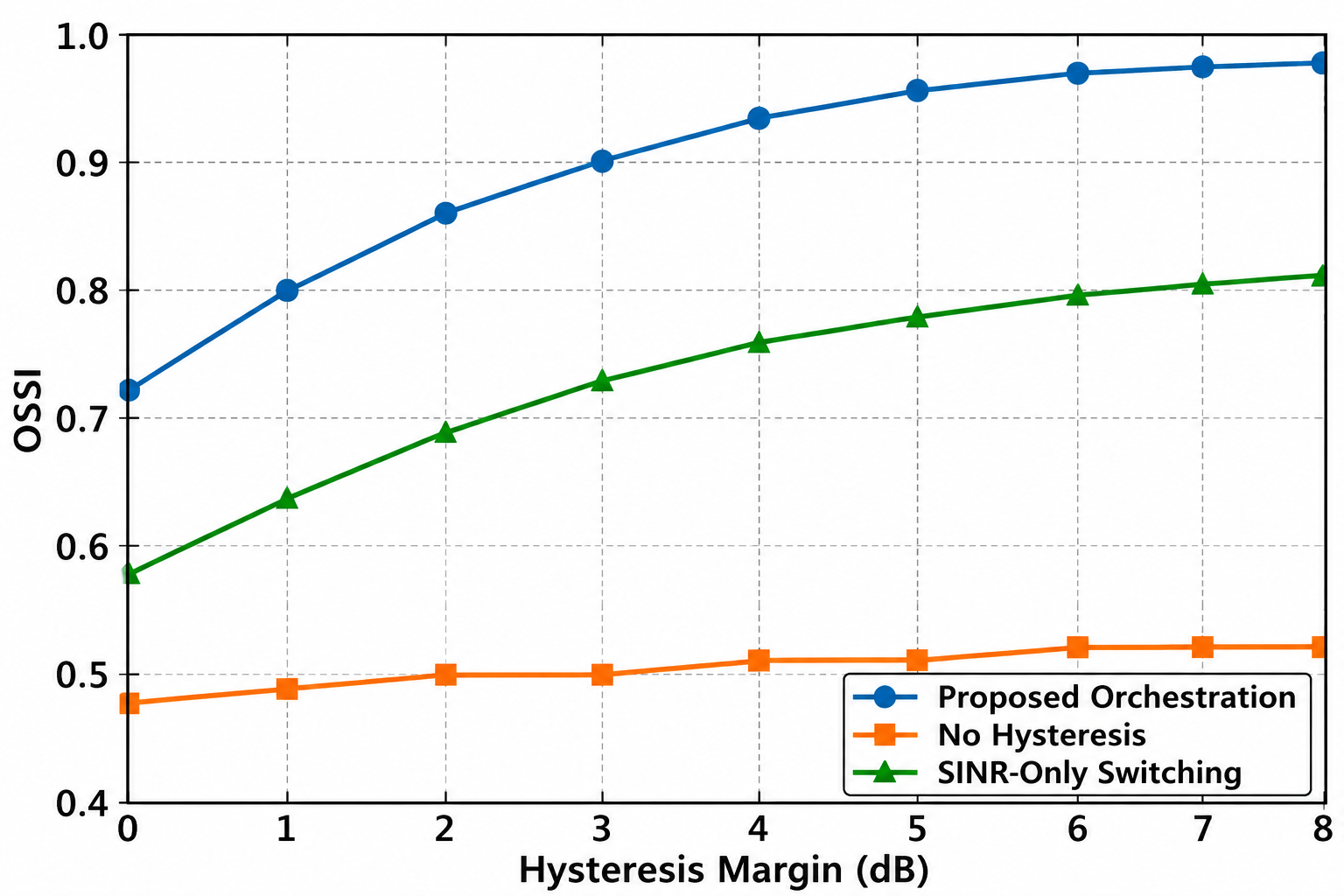}
\caption{OSSI versus hysteresis margin.}
\label{fig:ossi_hysteresis}
\vspace{-3mm}
\end{figure}
Fig.~\ref{fig:ossi_hysteresis} shows OSSI versus hysteresis margin,
i.e., the SINR threshold gap required before switching to a candidate
serving domain. As the hysteresis margin increases, fewer short-term
SINR fluctuations trigger serving-domain transitions, thereby reducing
$N_{\mathrm{unstable}}$ in eq.~\eqref{eq:ossi} and increasing OSSI.
This trend confirms that the hysteresis mechanism improves association
stability by filtering transient channel variations rather than
triggering unnecessary serving-domain changes. SA-MO-RSMA achieves the
highest OSSI because its service-aware selection mechanism jointly
accounts for communication quality, service continuity, and switching
stability. Consequently, it suppresses unnecessary domain transitions
more effectively than the benchmark schemes. In contrast, the
no-hysteresis and SINR-only baselines are more sensitive to short-term
SINR fluctuations and consequently exhibit lower switching stability.
\begin{figure}[t]
\centering
\includegraphics[width=0.82\columnwidth]{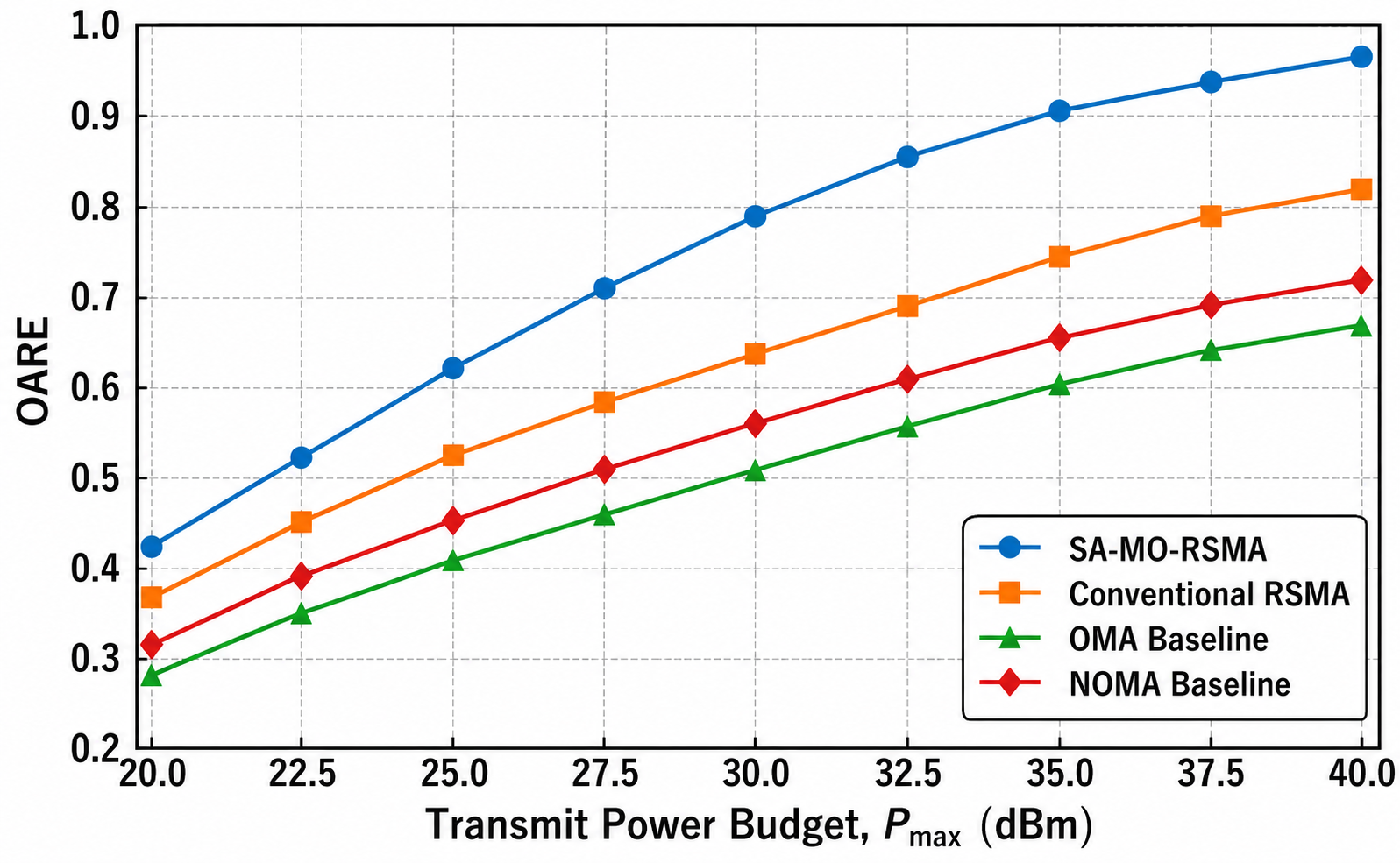}
\caption{OARE versus transmit-power budget.}
\label{fig:oare_power}
\vspace{-3mm}
\end{figure}

Fig.~\ref{fig:oare_power} shows OARE versus transmit-power budget.
Over the considered operating range, OARE increases for all schemes
because the achievable rate gains outweigh the corresponding increase
in transmit power, thereby improving the rate-to-power trade-off.
SA-MO-RSMA achieves the highest OARE by exploiting its global-common,
orbit-common, and user-private stream structure. Conventional RSMA
outperforms OMA and NOMA but remains below SA-MO-RSMA because it lacks
domain-specific common signaling. In SA-MO-RSMA, orbit-common streams
deliver shared information to users associated with the same domain,
while user-private streams carry user-specific data. This reduces
redundant transmission of shared information through private streams
and improves throughput per unit transmit power under heterogeneous
GEO, MEO, LEO, and terrestrial conditions.
\section{Conclusion}
\label{sec:conclusion}

This paper proposed SA-MO-RSMA for service-aware multi-orbit
orchestration in integrated 6G TN--NTN systems. The framework combines
hierarchical global-common, orbit-common, and user-private RSMA
transmission with service-intent-aware serving-domain selection and
OSUF-based orchestration. We formulated the joint serving-domain
selection and resource-allocation problem and introduced MOSCR, OSSI,
OARE, and the synchronization-quality metric $J$ to characterize
service continuity, switching stability, power-normalized transmission
efficiency, and synchronization robustness, respectively. Analytical
results established boundedness and solution-existence properties,
while numerical results demonstrated improved average per-user service
utility, continuity, switching stability, and transmission efficiency
over the considered benchmarks. Future work will investigate
distributed orchestration, learning-assisted adaptation, and alignment
with evolving 3GPP NR--NTN and ETSI ISG MAT activities.

\appendix
\section{Supporting Proofs}
\label{app:proofs}

\textit{Proof of Lemma 1:}
From eq.~\eqref{eq:orbit_utility}, the normalized components lie in
$[0,1]$, with $w_i\geq0$ and $\sum_{i=1}^{6}w_i=1$. Hence,
$0\leq\sum_{i=1}^{5}w_i\hat{X}_{i,u,o}\leq1-w_6$ and
$-w_6\leq-w_6\hat{L}_{u,o}\leq0$, yielding
$-w_6\leq\Psi_{u,o}\leq1-w_6$.

\textit{Proof of Proposition 1:}
For any user $u$, $\Psi_{u,o}$ is evaluated over the finite, nonempty
set $\mathcal O$. Hence, $\{\Psi_{u,o}:o\in\mathcal O\}$ contains a
maximum and at least one $o_u^\star$ exists; multiple maximizers may
exist.

\textit{Proof of Lemma 2:}
From eq.~\eqref{eq:osuf}, all normalized components lie in $[0,1]$,
with non-negative weights summing to unity. Hence, the benefit
contribution lies in $[0,1-w_L-w_O]$, while the cost contribution lies
in $[-(w_L+w_O),0]$. Therefore,
\[
-(w_L+w_O)\leq U_{u,o}^{\mathrm{OSUF}}
\leq1-w_L-w_O,
\]
which proves the boundedness of $U_{u,o}^{\mathrm{OSUF}}$.\\
\textit{Proof of Proposition 2:}
For fixed benefit components,
$\partial U_{u,o}^{\mathrm{OSUF}}/\partial\hat{L}_{u,o}=-w_L\leq0$
and
$\partial U_{u,o}^{\mathrm{OSUF}}/\partial\hat{O}_{u,o}=-w_O\leq0$.
Hence, increasing either cost cannot increase the OSUF.

\textit{Proof of Theorem 1:}
The association vector $\mathbf A$ belongs to the finite set
$\mathcal O^U$. For each feasible association,
$(\mathbf P,\mathbf\Theta)$ belongs to a compact feasible set under the
stated assumptions. Since $U_{\mathrm{net}}^{\mathrm{OSUF}}$ is
continuous, it attains a maximum for each feasible association.
Finiteness of $\mathcal O^U$ therefore guarantees at least one global
maximizer whenever the overall feasible set is nonempty.

\textit{Proof of Theorem 2:}
Consider $U_1<U_2$ under fixed resources and channel conditions. By
assumption, increasing user load does not increase the per-user benefit
components or decrease the corresponding costs. Since the benefit
weights are non-negative and the cost terms enter OSUF negatively,
$\bar{U}_{\mathrm{net}}^{\mathrm{OSUF}}(U_2)
\leq\bar{U}_{\mathrm{net}}^{\mathrm{OSUF}}(U_1)$.
Hence, the average per-user OSUF is non-increasing with user load.
\bibliographystyle{IEEEtran}
\bibliography{CSCN}

\end{document}